\documentclass[twocolumn,aps,showpacs,eqsecnum,10pt]{revtex4} 
\usepackage{latexsym}
\usepackage{amsthm}
\usepackage{amsmath}
\usepackage{amssymb}
\usepackage{epsfig}

\newtheorem{theorem}{Theorem}

\theoremstyle{definition}

\newcommand{\bra}[1]{\langle #1|}
\newcommand{\ket}[1]{| #1 \rangle }

\newcommand{\tr}[1]{{\rm tr}[#1]}
\newcommand{\Bra}{\langle}
\newcommand{\Ket}{\rangle}
\newcommand{\be}{\begin{eqnarray}}
\newcommand{\ee}{\end{eqnarray}}
\def\Tr{{\rm tr}}

\newcommand{\cE}{{\cal E}}
\newcommand{\cG}{{\cal G}}
\newcommand{\cI}{{\cal I}}

\newcommand{\cF}{{\cal F}}

\newcommand{\cH}{{\cal H}}

\newcommand{\cL}{{\cal L}}

\begin{document}
\title{Quantum Finite-Depth Memory Channels: Case Study}
\author{Tom\'a\v s Ryb\'ar$^{1}$, M\'ario Ziman$^{1,2}$}
\affiliation{
$^{1}$Research Center for Quantum Information, Institute of Physics, Slovak Academy of Sciences, D\'ubravsk\'a cesta 9, 845 11 Bratislava, Slovakia \\
$^{2}$Faculty of Informatics, Masaryk University, Botanick\'a 68a, 602 00 Brno, Czech Republic
}
\begin{abstract}
We analyze the depth of the memory of quantum memory channels generated by
a fixed unitary transformation describing the interaction between the principal
system and internal degrees of freedom of the process device.
We investigate the simplest case of a qubit memory channel 
with a two-level memory system. In particular, we explicitly characterize 
all interactions for which the memory depth is finite. We show that the 
memory effects are either infinite, or they disappear after at most two uses 
of the channel. Memory channels of finite depth can be to some extent
controlled and manipulated by so-called reset sequences. We show that
actions separated by the sequences of inputs of the length of the memory depth
are independent and constitute memoryless channels.
\end{abstract}
\pacs{03.65.Ta,03.65.Yz,03.67.Hk}
\maketitle

\section{Memory effects}
Schr\"odinger equation implies that an evolution of a closed quantum system 
is unitary. However, this ideal picture of closed and isolated quantum
system is very difficult to achieve experimentally. Unavoidable interactions 
between the system and its environment result in a nonunitary evolution.
Fortunately, under some specific though quite realistic conditions
the dynamics of the system can be described without the necessity 
of explicit consideration of the environment's degrees of freedom. 
The crucial assumption
of open system dynamics is that initially the system is statistically 
completely independent of the environment degrees of freedom 
affecting its time dynamics. It means that a preparation procedure is completely
uncorrelated from the evolution process. 

For example, a photon source (e.g. laser) is independent 
of an optical cable used for the transmission. Only after inserted into
the optical cable the photon is affected by its properties resulting
in a state change. Although the interaction between the photon
and the cable is driven by Schr\"odinger equation, the photon itself
undergoes a nonunitary evolution. In particular, let us denote
by $\varrho_1$ the initial state of the photon and by $\xi$ the initial
state of the environment represented by the optical cable. The input-output
transformation then reads
\be
\varrho_1\to\varrho_1^\prime={\rm Tr}_{\rm env}
[U\varrho_1\otimes\xi_{\rm env} U^\dagger]=\cE_1[\varrho_1]\, .
\ee
By definition the mapping $\cE$ describing the quantum process (channel)
is linear, completely positive and trace-preserving.

But, not only the photon state has changed. Also the environment degrees 
of freedom evolved into
\be
\xi^\prime_{\rm env} = {\rm Tr}_{1}
[U\varrho_1\otimes\xi_{\rm env} U^\dagger]=\cF[\xi_{\rm env}]\, .
\ee
This \emph{concurrent} mapping $\cF$ acting on the memory system
is a valid  channel, because it is linear, completely positive 
and trace preserving. Let us note that such concurrent channel 
depends only on the input system state, hence for any channel $\cE$ 
acting on a system there exist many concurrent channels $\cF$
acting on the memory, and {\it vice versa}.

If the same optical cable is 
used once more, then 
\be
\varrho^{\prime}_2=\cE_2[\varrho_2]={\rm Tr}[U\varrho_2\otimes
\xi_{\rm env}^\prime U^\dagger]\, ,
\ee
and $\cE_1\neq\cE_2$ in general. Moreover,
\be
\nonumber
\omega_{12}&=&{\rm Tr}_{\rm env}[U_2 U_1 
(\varrho_1\otimes\varrho_2\otimes\xi_{\rm env}) U_1^\dagger U_2^\dagger]\\
&\neq& \cE_1[\varrho_1]\otimes\cE_2[\varrho_2]\,,
\ee
where $U_1$ ($U_2$) acts on the environment and the first (second) system.
We see that subsequent usages of the same process device (e.g. optical cable) 
are not necessarily independent. Usually, a time intervals in between
the usages are sufficiently large so that the environment relaxes into 
its original initial state, hence $\omega_{12}
=\cE_1\otimes\cE_1[\varrho_1\otimes\varrho_2]$. If this holds for
any number of uses, we say that the device is memoryless and its action
can be fully described by means of quantum channels, i.e. completely positive
trace-preserving linear maps. However, our goal is to investigate 
the cases when the relaxation processes are not sufficiently fast 
(or are not happening at all) to guarantee the same conditions for 
each run of the experiment (e.g. photon transmission). Such
devices are described by quantum memory channels. In particular, we will
focus on characterization and properties of those memory channels, 
for which the memory effects are finite.

The research subject of quantum memory channels is relatively new. Once the
nature of the memory mechanism is known it can be exploited to increase 
the information transmission rates. Moreover, in this case the entangled
encoding strategies can significantly overcome the factorized ones.  Thus, 
the capacities (either classical, or quantum) of quantum memory channels 
are not necessarily additive. Naturally, the research is mostly focused on 
investigation of transmission rates for particular classes of memory channels 
\cite{macchiavello2002, macchiavello2004, ball2004, bowen2005, giovannetti2005a, karpov2006, karimipour2006, daems2007, datta2007, darrigo2007, caruso2008, wouters2009, dorlas2009}. Recently, attention has been paid to an interesting 
class of so-called bosonic memory channels
\cite{giovannetti2005b, cerf2005, ruggeri2005, ruggeri2007, pilyavets2008, lupo2009a, lupo2009b, schafer2009, lupo2009c} and also
to memory effects in the transmission of quantum states over the
spin chains \cite{plenio2007, bayat2008, plenio2008}. Our aim is to
investigate the structural properties of quantum memory channels rather
than to analyze their communication capabilities. A general framework
and structural theorem for quantum memory channels was given in the seminal work
of Kretschmann and Werner \cite{kretschmann2005}. In \cite{chiribella2008}
the discrimination of general quantum memory channels was investigated and
in \cite{rybar2008} the concept of repeatable quantum memory channels 
was introduced and analyzed. In \cite{kretschmann2005} the authors introduced
the concept of forgetful quantum memory channels and showed that these 
memory channels form a dense subset of all quantum memory channels. 
For such memory channels the state of the memory is ``forgotten'' after 
a certain number of uses. In other words, 
after $n$ uses of the memory channel 
the $(n+1)$th output state is approximately the same whatever was the 
original state of the memory. Our task is to identify those channels,
for which the output state is \emph{exactly} the same and to analyze the
the memory depth once the size of the memory system is fixed.

Let us note that the concept of finiteness of the memory we are going to use 
is different as the one introduced in Ref.\cite{bowen2004}, where 
the finiteness means the size of the memory
system. In our case, the finiteness is related rather to the depth of 
memory effects. Our ultimate goal is to clearly formulate
this concept and investigate the simplest case of qubit memory channels.
We want to characterize those memory channels for which the memory
depth is finite. Such memory channels can potentially mimic memoryless 
channels, paying the cost of larger inputs.

In the following Section II we will formalize the language of quantum 
memory channels. In Section III we will formulate the problem in general
settings. The qubit case will be investigated in details in Section IV.
The results are summarized in the last Section V. 

\section{Preliminaries}
Let us denote by $\cH$ a Hilbert space of the studied quantum system
and by $\cL(\cH)$ a set of bounded linear operators on $\cH$.
A state $\varrho$ is any positive linear operator on $\cH$ of unit trace, i.e.
$\varrho\ge 0$ and $\tr{\varrho}=1$. A linear map $\cE$ on the set 
of traceclass operators is called a channel
if it is completely positive [$\cL(\cH\otimes\cH_{\rm anc})\ni X\ge 0 $ 
implies $(\cE\otimes\cI_{\rm anc})[X]\ge 0$]
and trace-preserving ($\tr{\cE[X]}=\tr{X}$). The famous Stinespring dilation
theorem says that any channel can be realized as a unitary channel
on some extended Hilbert space, i.e.
\be
\cE[X]={\rm tr}_{\rm anc}[U (X\otimes\xi_{\rm anc}) U^\dagger]
\ee
for some unitary operator $U\in\cL(\cH\otimes\cH_{\rm anc})$
and some state $\xi_{\rm anc}$.

By a process device we will understand any fixed piece of hardware
transforming quantum system from their initial state to some final
state. In each individual use it is described by some quantum 
channel, i.e. $\varrho\mapsto\varrho^\prime=\cE[\varrho]$. It is 
memoryless if its joint action on $n$ subsequent inputs is 
factorized and in each run it is the same, i.e.
$\cE_{1\dots n}=\cE_1\otimes\cdots\otimes\cE_1$ for all 
$n=1,2,\dots$. If such property does not hold then no single
channel can be used to describe the quantum process device. The process
device is in general described by an infinite sequence
$\cE_1,\cE_{12},\dots$ of channel acting on 
$\cH, \cH\otimes\cH,\dots$, respectively. The causality requirement that 
the actual action does not depend on future inputs implies that 
\be
\nonumber
{\rm tr}_n\cE_{12\dots n}[X_{1,2,\dots,n-1}\otimes Y_n]
=\cE_{1,2,\dots, n-1}[X_{1,2,\dots,n-1}]
\ee 
for all $X,Y$. In the seminal work \cite{kretschmann2005} it was shown 
that such causal quantum memory channel can be always expressed
as a concatenation of unitary channels describing a sequence of
interactions between the individual inputs and some fixed memory 
system, i.e.
\be
\nonumber
\cE_{1,2,\dots n}[\omega_{12\dots n}]=
{\rm tr}_{\rm mem}[U_n\dots U_1(\omega_{12\dots n}\otimes \xi_{\rm mem}) 
U_1^\dagger\dots U_n^\dagger]\,,
\ee
where $\xi_{\rm mem}$ is a state of an ancillary system called memory
and the bipartite unitary operator $U_j$ acts nontrivially only 
on the $j$th input and the memory system. This representation is not unique
and by definition we assume that we do not have direct access to 
the memory system. 

In what follows we shall restrict to a specific type of quantum memory 
models, in which the interactions are described by the same unitary 
operator, i.e. $U_1=U_2=\cdots=U_n=U$. Let us note that 
for general considerations this case covers the most general situation. 
In particular, let $U_1,U_2,\dots$ be the sequence of unitaries defining 
a quantum memory channel (potentially $U_j\neq U_k$). We can 
define a unitary operator
$W=\sum_{j=0}^\infty U_{j}\otimes \ket{j+1}\bra{j}$ on $\cH\otimes\cH_{\rm mem}
\otimes\cH_\infty$, where $\cH_\infty$ is the Hilbert space of the 
linear harmonic oscillator (being part of the memory system) 
and $U_j$ are the unitaries associated with the quantum memory 
channel. In this sense any quantum memory channel is generated by a fixed 
unitary operator $U=W$ and some initial memory state $\xi_{\rm mem}$. 
However, such reduction requires infinite memory system.

Let us stress that only if the input states are uncorrelated,
$\omega_{12\dots n}=
\varrho_1\otimes\varrho_2\otimes\cdots\otimes\varrho_n$, then
the transformation of each input state is described by a quantum 
channel. Otherwise, the channel model is not applicable.
On one side this is indeed a restrictive condition, however, 
on the other side it is experimentally very relevant. 
The channel $\cE_n$ transforming the $n$th input, in general, depends
on all previous inputs $\varrho_1,\dots,\varrho_{n-1}$. If this is the case
for all $n$, then we say that the memory is infinite. The other extreme
is the memoryless case, when $U=V\otimes V_{\rm mem}$ and the 
channel $\cE_n$ is completely independent of any input. For example,
if $\cH_{\rm mem}\equiv\cH$ and $U=V_{\rm swap}$ is the swap operation 
($V_{\rm swap}\varrho\otimes\xi V_{\rm swap}^\dagger=\xi\otimes\varrho$), then 
$\cE_n[\varrho_n]=\varrho_{n-1}$, thus, $\cE_n$ is a complete contraction
of the state space into the state $\varrho_{j-1}$, which describes the
$(n-1)$th input. In such case the memory is of finite depth, because 
$\cE_n$ depends solely on the input state $\varrho_{n-1}$. 

In general, we say that a memory of the quantum memory channel 
generated by a unitary operator $U$ is of depth $\Delta_U$, 
if for each $n$ the channel $\cE_n$ does not depend on the 
initial memory state $\xi_{\rm mem}$, neither on the particular 
choice of input states $\varrho_j$ for all $j<n-\Delta_U$. 
Or, alternatively, the depth is $\Delta_U$ if for each 
$n$ the channel $\cE_n$ is independent of the inputs preceding 
$(n-\Delta_U)$th run of the process device including the original
memory state $\xi_{\rm mem}$. For example, the SWAP operator is of depth 
1, i.e. $\Delta_{V_{\rm swap}}=1$.

Our goal is to analyze which interactions $U$ generate
memory channels with finite memory irrespective of the initial state
of the memory system.

\section{Finite depth memory}
The channel $\cE_j$ transforming a given input $\varrho_j$
is generated by the interaction $U$ and the state of the 
ancilla $\xi_j$ in the $j$th run of the process device.
All the parameters the channel $\cE_j$ depends on are only mediated 
through the memory state $\xi_j$. Choosing an orthogonal 
operator basis $\tau_0,\dots,\tau_{d^2-1}$
of the memory system the memory state $\xi$ takes the form
\be
\xi=\sum_k m_k \tau_k\, ,
\ee
and the resulting channel reads
\be
\cE_\xi[\varrho]=\sum_k m_k {\rm tr}_{\rm mem}[U\varrho\otimes\tau_k U^\dagger]\,.
\ee
Let us note that orthogonality is defined with respect to Hilbert-Schmidt 
scalar product $\Bra A,B\Ket_{\rm hs}=\tr{A^\dagger B}$.

If for a fixed unitary operator $U$ and arbitrary
input state $\varrho$ we have 
${\rm tr}_{\rm mem}[U\varrho\otimes A U^\dagger]=O$ for some
operator $A$, then
the induced channels $\cE$ are independent of parameter 
$\tr{\xi A}$. It follows from the fact that the operator $A/\tr{A^\dagger A}$ 
can be taken to be an element of the orthonormal operator basis $\{\tau_k\}$ 
and $\xi=\sum_k \tr{\xi\tau_k}\tau_k$. The set of all such 
operators $A$ form a linear subspace of $\cL(\cH)$
and we call the corresponding state parameters
$\tr{\xi A}$ irrelevant, because $\cE$ does not depend on them.
Let us note that the identity operator $I$ is never irrelevant, i.e.
${\rm tr}_{\rm mem}[U(\varrho\otimes I) U^\dagger]\neq O$. Therefore, 
without loss of generality we can set $\tau_0=I/\sqrt{d}$ and, 
consequently, due to orthogonality the other elements of the operator 
basis are traceless, i.e. $\tr{\tau_k}=0$ for all $k\neq 0$. Thus, the 
irrelevant operators are necessarily traceless. In such basis the states $\xi$
take the form $\xi=\frac{1}{d}I+\vec{m}\cdot\vec{\tau}$, hence they
are uniquely represented by $(d^2-1)$-dimensional vectors $\vec{m}$ (so-called
Bloch vectors). The entries of each vector $\vec{m}$ can be split
into relevant and irrelevant ones. We will focus 
on the behavior of the relevant parameters mediating the memory effects.

Using the process device $n$ times the memory undergoes an evolution
\be
\xi_{n+1}=\cF_n[\xi_{n}]=\cdots = \cF_n\cdots\cF_1[\xi_{\rm mem}]\,,
\ee
where $\cF_j$ is defined via $\cF_j[\xi_{j}]={\rm tr_{sys}}[U\varrho_{j}
\otimes \xi_{j} U^\dagger]$ and $\xi_1=\xi_{\rm mem}$ is the initial state
of the memory system. Let us define a channel $\cG=\cF_n\cdots\cF_1$. This
channel potentially depends on all input states $\varrho_1,\dots,\varrho_n$,
hence, consequently, the memory state $\xi_{n+1}$ and also the channel $\cE_{n+1}$
depend on $\xi_1$ and all inputs $\varrho_1,\dots,\varrho_n$. If the memory
is finite and of the depth $n$, then $\cE_{n+1}$ does not depend on $\xi_1$
whatever collection of input states $\varrho_1,\dots,\varrho_n$ was used.
This happens if the relevant parameters of $\xi_{n+1}$ do not depend
on the memory state $\xi_1$. Let us note that $\xi_n$ still may depend
on input states $\varrho_1,\dots,\varrho_n$, however, it is independent 
on any input preceding $\varrho_1$. As it is required this feature is invariant 
in time. That is, $\cE_{s+n+1}$ is independent of memory state
$\xi_s$ and also on all input states $\varrho_j$ with $j\leq s$. 

The goal is to investigate for which $n$ the concurrent 
channel $\cG$ is deleting 
all relevant parameters of the memory system 
whatever sequence $\varrho_1,\dots,\varrho_n$ is used. The action of the
channel $\cG$ on Bloch vectors $\vec{m}$ takes the form of an affine
mapping, i.e. $\vec{m}\mapsto \vec{g}+G\vec{m}$, where
$g_k=\frac{1}{d}\tr{\tau_k^\dagger\cG[I]}$ and $G_{kl}=\tr{\tau_k^\dagger\cG[\tau_l]}$
for $k,l=1.\dots,d^2-1$. Since $\cG$ is a composition of channels
$\cF_1,\dots,\cF_n$ , using the corresponding 
vectors $\vec{f}_j$ and matrices $F_j$, the action can be expressed as
\be
\nonumber
\vec{m}_1\to \vec{m}_{n+1}=(F_n\cdots F_1)\vec{m}+(F_n\cdots F_2)\vec{f}_1
+\cdots +\vec{f}_n\,,
\ee
thus $G=F_n\cdots F_1$ and $\vec{g}=(F_n\cdots F_2)\vec{f}_1+\cdots +\vec{f}_n$.
The requirement of finite depth of the memory implies that
relevant parameters of $\vec{m}_{n+1}$ are independent of $\vec{m}_1$
for all input states $\varrho_1,\dots,\varrho_n$, 
hence, $G$ is singular and maps any vector $\vec{m}_1$ into the subspace 
spanned by ``irrelevant'' operators $\tau_k$. Let us note that product 
of nonsingular matrices is not singular. Since we do require that $G$ is 
singular for all sequences of inputs it follows that each $F_j$ must be 
singular. If for some input state $\varrho$ the matrix $F$ is not singular,
then sequence $\varrho^{\otimes n}$ induces a nonsingular matrix $G=F^n$
for arbitrary $n$. In such case, the memory depth is infinite.
Therefore, the singularity of the matrices $F$ for all input states
$\varrho$ is a necessary (but not sufficient) condition for $U$ to generate a 
finite quantum memory channel.

Let us note that a finite depth memory channel does not create any 
correlations between outputs separated by $n$ uses if all inputs are 
factorized (see Appendix A). Consequently, its actions (separated by $n$ uses) 
are independent. In this way the memory process device can be used to 
implement a memoryless channel, using first $n$ inputs as a reset sequence 
which will set the memory system to some particular (although not arbitrary) 
state ignoring the outputs and then performing the channel on next input. 
The proof of this statement is given in appendix A.

\section{Case study: two-dimensional memory}
In this section we will investigate qubit memory channels with
a two-dimensional memory system. The question is what are the possible values
of $\Delta$ in such very specific settings. Let us use the basis of Pauli 
operators $\sigma_x,\sigma_y,\sigma_z$
to express the qubit states. Then
the memory state takes the form 
$\xi_1=\frac{1}{2}(I+\vec{m}_1\cdot\vec{\sigma})$ and can be represented
by a three-dimensional Bloch vector $\vec{m}_1$. Similarly, let us assume
that the system is initially prepared in a state $\varrho_1=
\frac{1}{2}(I+\vec{r}_1\cdot\vec{\sigma})$. The action of the concurrent channel
$\cF_1[\xi_1]={\rm tr}_{\rm sys}[U\varrho_1\otimes\xi_1 U^\dagger]$ can be 
expressed by means of vector $\vec{f}_1=\frac{1}{2}\tr{\vec{\sigma}
\cF[I]}$
and matrix $F_{1,jk}=\frac{1}{2}\tr{\sigma_j\cF_1[\sigma_k]}$.
In particular, in the language of Bloch vectors the channel
takes an affine form $\vec{m_1}\to\vec{t}+T\vec{m_1}$, hence,
in the  $n$th run the memory system is transformed as 
$\vec{m}_n\to\vec{f}_n+F_n\vec{m}_n$, where by $\vec{m}_n$
we denoted the state of the memory before the $n$th use of the process 
device. As before, the initial memory $\vec{m_1}$ is transformed
as follows
\be
\nonumber
\vec{m}_1\to \vec{m}_{n+1}=(F_n\cdots F_1)\vec{m_1}+(F_n\cdots F_2)\vec{f}_1
+\cdots +\vec{f}_n\,.
\ee

A general two-qubit unitary transformation can be expressed as follows
(see for example \cite{kraus2000})
\be
U=(V_1\otimes W_1) e^{i\sum_j \alpha_j\sigma_j\otimes\sigma_j} (V_2\otimes W_2)\, ,
\label{eq_Unitary}
\ee
where $V_j,W_j$ are single qubit unitary operators and $\alpha_j$ 
are real numbers. We learnt that in order to generate a quantum memory
channel with finite depth of the memory for all input sequences, 
it is necessary for $U$ that the induced concurrent 
channels $\cF_j$ are singular.
Since local unitary rotations $V_j\otimes W_j$ do not affect 
the singularity, it is sufficient for now to analyze only the unitary operators 
of the form $U=e^{i\sum_j\alpha_j\sigma_j\otimes\sigma_j}$.

For the considered unitary operator $U=e^{i\sum_j\alpha_j\sigma_j\otimes\sigma_j}$
the matrix $F$ takes the form
\begin{equation}
F(\vec{r})=\left(
\begin{array}{ccc}
c_y c_z & r_z c_y s_z & - r_y s_y c_z \\
-r_z c_x s_z & c_x c_z & r_x s_x c_z \\
r_y c_x s_y & -r_x s_x c_y & c_x c_y
\end{array}
\right)\,,
\end{equation}
where $c_j=\cos{2\alpha_j}$ and $s_j=\sin{2\alpha_j}$. Let us note that
due to symmetry of $U$ with respect to exchange of the 
system and the memory, the same matrix describes the channel
acting on the system, only the role of $\vec{r}$ is replaced
by the initial state of the memory $\vec{m_1}$.

Evaluating
the determinant we get
\be\nonumber
\det F(\vec{r})= r_x^2 s_x^2 c_y^2 c_z^2 + 
r_y^2 c_x^2 s_y^2 c_z^2 +
r_z^2 c_x^2 c_y^2 s_z^2 +c_x^2 c_y^2 c_z^2  \,.
\ee
It vanishes if and only if at least one of the following
conditions hold
\be
\cos{2\alpha_x}=\cos{2\alpha_y}=0\,;\\
\cos{2\alpha_x}=\cos{2\alpha_z}=0\,;\\
\cos{2\alpha_y}=\cos{2\alpha_z}=0\,.
\ee
If exactly one of the above conditions holds, for instance 
$\cos{2\alpha_x}=\cos{2\alpha_y}=0$, then
\begin{equation}
F(\vec{r})=(\cos{2\alpha_z})\left(
\begin{array}{ccc}
0 & 0 & \pm r_y\\
0 & 0 & \pm r_x\\
0 & 0 & 0
\end{array}
\right)\,,
\label{eq_Fr}
\end{equation}
is a matrix of rank one and $F(\vec{r}_2) F(\vec{r}_1)=O$. 
Setting $F_j=F(\vec{r}_j)$ we get for all $j$
\be
\vec{m}_{j+1}=F_j\vec{m}_j+\vec{f}_j=F_{j}\vec{f}_{j-1}+\vec{f}_j\, .
\ee
Since $\vec{f}_j$ depends only on input state $\varrho_j$
the state of the memory $\xi_{j+1}$ depends only
on input state $\varrho_j$ and $\varrho_{j-1}$, i.e. on preceding
two input states. Therefore, the memory depth equals $\Delta=2$.
That is, the $j$th input state is transformed by a channel $\cE_j$
\be
\vec{r}_j^\prime = E_j\vec{r}_j + \vec{e}_j\,,
\ee
where $E_j,\vec{e}_j$ depends via the memory state $\vec{m}_j$
on input states $\varrho_{j-1}$ and $\varrho_{j-2}$.

Due to already mentioned symmetry of $U$ 
it follows that the channel $\cE_\xi$ acting on the system qubit
does not depend on the value of $m_z$, because
\be
E=(\cos{2\alpha_z})\left(
\begin{array}{ccc}
0 & 0 & \pm m_y\\
0 & 0 & \pm m_x\\
0 & 0 & 0
\end{array}
\right)\,. 
\label{eq_kanal}
\ee
The unitary operators $U=\exp{(i\sum_j \alpha_j\sigma_j\otimes\sigma_j)}$ 
generating the considered finite memory channels 
($\alpha_x,\alpha_y\in\{\pm\pi/4\}$) 
are of the form 
\be
\label{eq:Ualfa}
U_{\alpha_z}=
\frac{1}{2}[I+\sigma_{zz}+ie^{-2i\alpha_z}(\sigma_{xx}+\sigma_{yy})]
\sigma_{xx}^{h_x}\sigma_{yy}^{h_y}\,,
\ee
where $\sigma_{jj}=\sigma_j\otimes\sigma_j$, 
$h_j=H(-\alpha_j)$ ($j=x,y,z$) and $H(\cdot)$ is the Heavyside step function. 
The remaining options $\alpha_x,\alpha_z\in\{\pm \pi/4\}$ and
$\alpha_y,\alpha_z\in\{\pm\pi/4\}$ correspond to unitary operators that are
locally unitarily equivalent to $U_{\alpha_z}$. In particular, it is sufficient
to relabel the basis, i.e. instead of using the eigenbasis of $\sigma_z$
we use eigenbasis of $\sigma_x$, or $\sigma_y$ in which the unitary 
transformations $U_{\alpha_x},U_{\alpha_y}$ takes the same form.

The freedom as specified in \eqref{eq_Unitary} is a bit larger 
than that. Replacing the unitary operator $U_{\alpha_z}$ by a more
general one $U=V_1\otimes W_1 U_{\alpha_z} V_2\otimes W_2$ the concurrent
channel $F(\vec{r})$ takes the form
\be
F^\prime(\vec{r})=S^{\prime}F(\vec{r})R,
\ee
where $S^\prime$ and $R$ are orthonormal matrices corresponding 
to unitary operators $W_1$ and $W_2$, respectively. 
Since orthogonal matrices do not affect the singularity, 
the matrices $F^\prime(\vec{r})$ are singular.
Moreover, it can be rewritten in a more 
convenient form as $R^{-1}SF(\vec{r})R$, where $S=RS^\prime$ is a suitable 
orthogonal matrix. Using a sequence of input states 
$\varrho_1\otimes\cdots\otimes\varrho_n$ and defining 
$F^\prime_j=F^\prime(\vec{r}_j)$ we get
\be
G^\prime=F^\prime_n\cdots F^\prime_1=R^{-1}SF_n\cdots SF_1 R\, .
\ee
The question is for which values of $n$ and for which rotations
$S$ the matrices $G^\prime$ (generated by sequences $F_1,\dots,F_n$) 
maps memory states into the irrelevant subspace.

The matrix $R$ corresponds merely to changing the basis of memory system 
and as such does not affect the depth of memory of the memory channel 
and can be left arbitrary. We will not consider it in further calculations. 
The unitary matrix $W^{\prime}=W_2W_1$ corresponding to $S$ does not change 
the relevance of parameters, because for all operators $\tau$ 
and arbitrary $U$
\be
&&\Tr_{\rm mem}[(I\otimes W^{\prime}) U(\varrho\otimes\tau) U^{\dagger}(I\otimes W^{\prime\dagger})]=\nonumber\\
&&\sum_{abcd}\Tr[W^{\prime} \ket{a}\bra{b}\tau\ket{d}\bra{c}W^{\prime\dagger}]A_{ab}\varrho A^{\dagger}_{cd}=\nonumber\\
&&\sum_{abcd}\Tr[\ket{a}\bra{b}\tau\ket{d}\bra{c}]A_{ab}\varrho A^{\dagger}_{cd}=\nonumber\\
&&\Tr_{\rm mem}[U\varrho\otimes\tau U^{\dagger}],
\ee
where we used the expression $U=\sum_{ab}A_{ab}\otimes\ket{a}\bra{b}$
for some orthonormal basis $\{\ket{a}\}$ and operators $A_{ab}$
such that $U$ is unitary.

As we have seen in Eq.~\eqref{eq_kanal} there is only one irrelevant parameter 
$m_z$, because only $m_z$ does not enter the expression in Eq.~\eqref{eq_kanal}.
Consequently, we require for all sequences $F_1,\dots,F_n$ the 
following conditions 
\be
SF_nS\ldots SF_1=
\left(
\begin{array}{ccc}
0 & 0 & 0 \\
0 & 0 & 0 \\
x_1 & x_2 & x_3 \\
\end{array}
\right)\,,
\label{eq_cond}
\ee
where $x_1,x_2,x_3$ are arbitrary numbers, and $n$ will be the depth 
of this channel. 

Let us denote by $S_{kl}$ the entries of $S$ and define 
$a_{j,k}=\cos(2\alpha_z)(\pm r_{j,y}S_{k1}\pm r_{j,x}S_{k2})$ with
$j=1,\dots, n$ and $k,l=1,2,3$.
Then
\begin{equation}
SF_j=\left(
\begin{array}{ccc}
0 & 0 & a_{j,1}\\
0 & 0 & a_{j,2}\\
0 & 0 & a_{j,3}
\end{array}
\right)
\end{equation}
and the Eq.\eqref{eq_cond} reads
\be\nonumber
a_{1,3}\dots a_{n-1,3}
\left(
\begin{array}{ccc}
0 & 0 & a_{n,1} \\
0 & 0 & a_{n,2} \\
0 & 0 & a_{n,3} \\
\end{array}
\right)=
\left(
\begin{array}{ccc}
0 & 0 & 0 \\
0 & 0 & 0 \\
x_1 & x_2 & x_3 \\
\end{array}
\right)\,.
\ee
Since this relation must hold
for all states $\varrho$, i.e. for all Bloch vectors 
$\vec{r}_1, \dots,\vec{r}_n$, it is necessary that
$a_{j,3}=\pm r_{j,y} S_{31}\pm r_{j,x} S_{32}=0$ for
all vectors $\vec{r}_j$, thus, $S_{31},S_{32}=0$. Rotation matrices $S$
satisfying such constraint are necessarily of the form
\be
S=
\left(
\begin{array}{ccc}
q\cos2\beta & q\sin2\beta & 0 \\
-\sin2\beta &\cos2\beta & 0 \\
0 & 0 & q
\end{array} \right)\,,
\ee
where $q=\pm 1$ and $\beta\in[0,2\pi]$.
Therefore,
\be
SF_j=
\left(
\begin{array}{ccc}
0 & 0 &  q(r_{j,y}\cos2\beta  +r_{j,x}\sin2\beta)\\
0 & 0 & - r_{j,y}\sin2\beta +r_{j,x}\cos2\beta\\
0 & 0 & 0
\end{array}
\right)
\ee
are matrices of the same form as for $F_j$ only. The same
arguments imply that the depth is either 1, or 2, because
$SF_2 SF_1=O$ for all possible matrices $F_1,F_2$.
The unitary matrix $W^\prime$ corresponding to $S$ equals to
\be
\label{eq:Wprime}
W^{\prime}=\left(
\begin{array}{cc}
0 & 1 \\
-1 & 0
\end{array} \right)^{\frac{1-q}{2}}
\left(
\begin{array}{cc}
e^{i\beta} & 0 \\
0 & e^{-i\beta}
\end{array}\right).
\ee

In conclusion, the memory is finite only 
if the quantum memory channel is induced by unitary operator $U$
of the form (in some factorized basis)
\be
U=
(V_1\otimes W_2^{\dagger}W^{\prime}) U_{\alpha} (V_2\otimes W_2)\,.
\ee
Moreover, in such case necessarily $\Delta_U\leq 2$, hence, the memory 
depth (if not infinite) is surprisingly quite limited. If 
$\alpha_x=\alpha_y=\alpha_z=\pi/4$,
then $F_j$ is a zero matrix, $F_j\equiv O$, and $U_{\pi/4}=e^{i\pi/4\sum_j 
\sigma_j\otimes\sigma_j}=V_{\rm swap}$ is the swap operator. 
In such case, 
\be
\vec{m}_j&\mapsto&\vec{m}_{j+1}=\vec{f}_j=\vec{r}_j\, ;\\
\vec{r}_j&\mapsto&\vec{r}_{j}^\prime =\vec{m}_j=\vec{r}_{j-1}\, ;
\ee
where $j$th output state equals to $(j-1)$th input state, 
i.e. $\Delta_{V_{\rm swap}}=1$. In summary, the depth of 
the memory $\Delta_U$ in the considered case of single qubit 
memory systems can achieve only the values 0,1,2, or infinity.

\subsection{Classical bits}
Let us shortly discuss the case of classical memory channels. 
Quantum description covers the classical one in a sense that classical
states are density operators orthogonal in some fixed (factorized) basis, i.e. 
they represent probability distributions expressed as diagonal matrices.
Similarly, unitary operators are replaced by permutations, which form
a very specific subgroup of all unitary operators. Having in mind these
restrictions all the discussed concepts are applicable for classical 
systems as well. 

A classical bit is the simplest classical system having the quantum bit
as its quantum counterpart. The states are expressed as density operators 
$p\ket{0}\bra{0}+(1-p)\ket{1}\bra{1}$ and there are only two permutations
corresponding to $I$ and $\sigma_x$, which flips the bit values. 
Assuming the memory system is also of the size of a single classical bit, 
there are only 4!=24 permutations $U$ describing the classical memory 
channels of a single bit. Analyzing all of them we find that
the memory depth can be 0,1, or infinity, because $U_\alpha$ describes
a permutation only if $\alpha=\pi/4$, i.e. when it is the SWAP operator.

\section{Conclusion}
For each quantum memory channel describing any quantum process device
we can assign a parameter $\Delta_U$ meaning that its $n$th run depends 
at most on the previous $\Delta_U$ uses. Equivalently, the input-output 
action is irrelevant of the state of the memory after the $(n-\Delta_U)$th use. 
We call this number the depth of the memory. We investigated in details 
the simplest case of qubit memory channels with the memory system composed 
of a single qubit, as well. We showed that values of the memory depth are 
restricted and $\Delta_U\in\{0,1,2,\infty\}$. Let us note that 
in the analogous situation for classical systems 
$\Delta_U\in\{0,1,\infty\}$. In particular,
$\Delta_U=0$ if $U$ is factorized, $\Delta_U=1$ if $U$ is the SWAP 
operator (up to local unitaries) and $\Delta_U=2$ if 
\be
\nonumber
U=\ket{\varphi}\bra{\varphi^\prime}&\otimes&\ket{\psi}\bra{\psi^\prime}+
\ket{\varphi_\perp}\bra{\varphi_\perp^\prime}\otimes\ket{\psi_\perp}\bra{\psi^\prime_\perp}+\\
\nonumber 
+ie^{-2i\alpha} (
\ket{\varphi}\bra{\varphi^\prime_\perp}&\otimes&\ket{\psi_\perp}\bra{\psi^\prime}+
\ket{\varphi_\perp}\bra{\varphi^\prime}\otimes\ket{\psi}\bra{\psi^\prime_\perp})
\,,
\ee
where $\ket{\psi}=W_2^\dagger W^\prime\ket{0}$ (see Eq.\eqref{eq:Wprime}), 
$\ket{\psi^\prime}=W_2^\dagger\sigma_{xx}^{h_x}\sigma_{yy}^{h_y}\ket{0}$,
$\ket{\varphi}=V_1\ket{0}$, $\ket{\varphi^\prime}=V_2^\dagger\sigma_{xx}^{h_x}\sigma_{yy}^{h_y}\ket{0}$.
In all other cases the memory is infinite.

If the memory depth is finite, then a sequence of input states
can be used to reset the memory system into a fixed state irrelevant
of the initial state of the memory and inputs preceding the reset 
input sequence. Applying the same reset sequence guarantees
that in each $(\Delta_U+1)$th use locally the same channel is 
implemented. In Appendix it is shown that actions of the process
device separated by reset sequences are indeed uncorrelated.

That is, in each $(n+1)$th run of the process device the same
quantum channel is independently implemented providing that the same reset
sequence is used. In this way, memory channels can be used as memoryless
ones. However, that it is an open problem whether any 
channel can be implemented on some finite-depth memory channel in this way
and also whether there is some bound on the size of the reset sequence 
and the memory system. So far, we know that if we restrict ourselves to
single qubit memory, then such channels are represented by rank-1
matrices and the reset sequence is of length at most 2. 

In summary, for most of the qubit memory channels the memory effects have 
infinite depth. Based on our investigation of the simplest physical
model we can make a rather surprising conjecture that the dimension 
of the memory puts constraints on the memory depth $\Delta_U$. Unfortunately,
we have not succeeded to find any simple analytic bound expressing this 
relation. Similarly, the characterization of general unitary operators
generating fine-depth memory channels remains open. 

\section*{Acknowledgments}
We acknowledge financial support via the European Union project 
HIP FP7-ICT-2007-C-221889, and via the projects
APVV LPP-0264-07 QWOSSI, VEGA-2/0092/09, OP CE QUTE ITMS NFP 262401022, 
and CE-SAS  QUTE.

\begin{appendix}
\section{Correlations}
\begin{theorem}
Consider a unitary memory channel $U$ of the depth $n$, i.e.
$\Delta_U=n$. Then the actions of the process device separated by $n$ 
uses (reset sequence) are not correlated providing that the reset
sequences are not correlated, i.e.
\be
\cE[\omega_{n+1,2n+2}]=(\cE_{n+1}\otimes\cE_{2n+2})[\omega_{n+1,2n+2}]\,,
\ee
where $\omega_{n+1,2n+2}$ is the joint state of $(n+1)$th
and $2(n+1)$th inputs and $\cE_j$ denotes the action of the
memory channel on its $j$th input.
\end{theorem}

\begin{proof}

Let us denote by $\Xi=R_1\otimes\cdots\otimes R_n$ 
the sequence of input states forming the so-called reset sequence. 
This sequence, together with the memory system $\xi$, is inducing a channel $\cE_{\Xi}^{\xi}$ on the $(n+1)$th process device input state
\be
\cE_{\Xi}^{\xi}[\omega]=\Tr_{\rm res, mem}
[U^{(n+1)}(\Xi\otimes\omega\otimes\xi)U^{(n+1)\dagger}]\,,
\ee 
where $U^{(n+1)}$ is the $n+1$-fold concatenation of the channel $U$
and $\omega$ is the state of input system. 
Let us express the interaction $U$ as follows
\be
U=\sum_{a,b}A_{ab}\otimes\ket{a}\bra{b}\,,
\ee
where $A_{ab}$ are operators acting on the principal system
and vectors $\{\ket{a}\}$ form an orthonormal basis of the Hilbert space
of the memory system. The unitarity of $U$ imposes 
the following normalization conditions on operators $A_{ab}$
\be
\sum_a A^\dagger_{ab} A_{ac}^{}=\delta_{bc}I\,, 
\quad
\sum_b A_{ab}^{} A^\dagger_{cb}=\delta_{ac}I\,. 
\ee
Defining the operators
\be
M_{a_n a_0}= \sum_{a_1,\dots,a_{n-1}} A_{a_1 a_0}\otimes\cdots\otimes A_{a_n a_{n-1}}
\ee
acting on the Hilbert space of the reset sequence $\cH_{\rm res}$
we get
\be
\nonumber
\cE_\Xi^\xi[\omega]&=&\hskip-0.3cm\sum_{\substack{a_0, a_n, a_{n+1}\\c_0,c_n}} 
\hskip-0.3cm\xi_{a_0 c_0}\tr{M_{a_n a_0}\Xi M_{c_n c_0}^\dagger}
A_{a_{n+1} a_n}\omega A^\dagger_{a_{n+1} c_n}\\
&=& \sum_{a,c} \xi_{ac}\Omega_{ac}(\Xi,\omega)
\,,
\ee
where $\xi_{ac}=\bra{a}\xi\ket{c}$, $A_{a_{j} a_{j-1}}$ acts
on $j$th input of the reset sequence and $\Omega_{ac}(\Xi,\omega)$ are 
operators defined on the $(n+1)$th principal system. These operators
depend on $\Xi,\omega$, but not on the state $\xi$. 

Then, the finite memory 
depth condition implies that for all memory 
states $\xi,\xi^\prime$ following relation holds
\be
\label{eq:a2}
\cE_{\Xi}^{\xi}[\omega]=\cE_{\Xi}^{\xi^\prime}[\omega]\equiv\cE_\Xi[\omega]\, ,
\ee
for all states $\omega$. Especially, for memory states
$\xi=\ket{a}\bra{a}$ we get $\cE_\Xi^{\ket{a}\bra{a}}[\omega]=
\Omega_{aa}(\Xi,\omega)=\Omega_0(\Xi,\omega)$ for all values of $a$. 
Using a general state $\Xi$ we obtain
\be
\cE_\Xi^\xi[\omega]=\Omega_0(\Xi,\omega)+
\sum_{a\neq c} \xi_{ac}\Omega_{ac}(\Xi,\omega)\,,
\ee
and, consequently, the condition \eqref{eq:a2} implies
that $\Omega_{ac}(\Xi,\omega)=O$ for all $a\neq c$. In summary,
\be
\nonumber
\Omega_{ac}(\Xi,\omega)&=&\sum_{a_n,a_{n+1},c_n}
\tr{M_{a_n a}\Xi M_{c_n c}^\dagger}A_{a_{n+1} a_n}\omega A^\dagger_{a_{n+1} c_n}\\
\label{eq:a3}
&=&\delta_{ac} \Omega_0(\Xi,\omega)\,,
\ee
and
\be
\cE_\Xi[\omega]=\Omega_0(\Xi,\omega)\,.
\ee

Next we add another reset sequence $\Xi_2$ followed by next input $\omega_2$
and analyze the joint action of the finite-depth memory process device on
the inputs $\omega_1$ and $\omega_2$. In such case
\be
\nonumber
& & \cE^\xi_{\Xi_1\otimes\Xi_2}[\omega_1\otimes\omega_2]=\\
\nonumber
& & =\Tr_{\rm res,mem}
[U^{(2n+2)}(\Xi_1\otimes\Xi_2\otimes\omega_{12}\otimes \xi)\nonumber 
U^{(2n+2)\dagger}]
\\ \nonumber & & 
= \sum \xi_{a_0 c_0}
\tr{M_{a_n a_0}\Xi_1M^\dagger_{c_n c_0}} 
A_{a_{n+1} a_n}\omega_1 A^\dagger_{c_{n+1},c_n}\otimes
\\ \nonumber & &
\tr{M_{a_{2n+1} a_{n+1}}\Xi_2M^\dagger_{c_{2n+1} c_{n+1}}} 
A_{a_{2n+2} a_{2n+1}}\omega_2 A^\dagger_{a_{2n+2},c_{2n+1}} 
\\ \nonumber & &
= \sum \xi_{a_0 c_0} 
\tr{M_{a_n a_0}\Xi_1M^\dagger_{c_n c_0}} 
A_{a_{n+1} a_n}\omega_1 A^\dagger_{c_{n+1},c_n}
\\ \nonumber & & 
\hspace*{1cm}\otimes\, \delta_{a_{n+1},c_{n+1}}\Omega_0(\Xi_2,\omega_2)
\\ \nonumber & &
=\Omega_0(\Xi_1,\omega_1)\otimes\Omega_0(\Xi_2,\omega_2)
\\ \nonumber & &
=(\cE_{\Xi_1}\otimes\cE_{\Xi_2})[\omega_1\otimes\omega_2]\,.
\ee
where we have used twice the identity in Eq.\eqref{eq:a3}. 
Let us note that due to linearity the inputs (separated
by the reset sequence $\Xi_2$) does not have to be factorized and altogether
are described by a density operator $\omega_{12}$. In conclusion, 
the actions separated by reset sequences take the ``memoryless'' form
\be
\cE_{\Xi_1\otimes\Xi_2}=\cE_{\Xi_1}\otimes\cE_{\Xi_2}\, .
\ee
This completes the proof.
\end{proof}
\end{appendix}

\end{document}